\documentclass[10pt]{article}

\usepackage[cmex10]{amsmath}
\usepackage{array}
\usepackage{fixltx2e}
\usepackage{stfloats}

\usepackage{amsfonts}%
\usepackage{amssymb}%
\usepackage{graphicx}
\usepackage[latin1]{inputenc}
\usepackage{hyperref}

\newtheorem{theorem}{Theorem}

\newtheorem{corollary}[theorem]{Corollary}

\newtheorem{definition}[theorem]{Definition}
\newtheorem{example}[theorem]{Example}

\newtheorem{lemma}[theorem]{Lemma}

\newtheorem{proposition}[theorem]{Proposition}
\newtheorem{remark}[theorem]{Remark}
\newtheorem{remarks}[theorem]{Remarks}

\newenvironment{proof}[1][Proof]{\textbf{#1.} }{\ \rule{0.5em}{0.5em}}

\newcommand{\N}{{\mathbb N}}
\newcommand{\Z}{{\mathbb Z}}
\newcommand{\F}{\mathbb F_q}
\newcommand{\B}{{\mathbb B}}
\renewcommand{\bf}[1]{\mathbf{#1}}
\newcommand{\Le}{\mathbb L}

\newcommand{\Q}{{\mathcal Q}}
\newcommand{\D}{{\mathcal D}}
\newcommand{\mad}{\mathrm{-mad}}

\newcommand{\tq}{\; \mid \;}
 
\newcommand{\supp}{{\rm supp}}   

\usepackage{color}

\begin{document}
\title{Ds-bounds for cyclic codes: New bounds for abelian codes.}

 \author{ J.J. Bernal, M. Guerreiro\footnote{M. Guerreiro is with Departamento de Matem\'atica, Universidade Federal de Vi\c cosa, 36570-900 Vi\c cosa-MG, Brazil. Supported by CNPq-Brazil.
E-mail: marines@ufv.br}  
and J. J. Sim\'on
\footnote{J. J. Bernal and J. J. Sim\'on are with the Departamento de Matem\'aticas, Universidad de Murcia, 30100 Murcia, Spain. Partially supported by MINECO project MTM2012-35240 and Fundaci\'{o}n S\'{e}neca of Murcia.
E-mail: \{josejoaquin.bernal, jsimon\}@um.es}
}

%



\maketitle


%
\begin{abstract}
In this paper we develop a technique to extend any bound for cyclic codes constructed from its defining sets (ds-bounds) to abelian (or multivariate) codes. We use this technique to improve the searching of new bounds for abelian codes. 
\end{abstract}

\section{Introduction}

\noindent 

The study of abelian codes is an important topic in Coding Theory, having an extensive literature, because they have special algebraic properties that allow one to construct good codes with efficient encoding and decoding algorithms. More precisely, regarding decoding, the two most known general techniques are permutation decoding \cite{BSpermdec} and the so called locator decoding~\cite{Blah} that uses the Berlekamp-Massey algorithm~\cite{Sakata} (see also~\cite{KZ}).

Even though the mentioned decoding methods require to know the minimum distance, or a bound for it, there are not much literature or studies on its computation and properties, or it does exist only for specific families of codes (see~\cite{Blah}). Concerning BCH bound, in \cite{Camion}, Camion introduces an extension from cyclic to abelian codes which is computed through the apparent distance of such codes.  Since then there have been some papers improving the original computation and giving a notion of multivariate BCH bound and codes (see~\cite{BBCS2,  Evans}).

These advances lead us to some natural questions about the extension to the multivariate case of all generalizations and improvements of the BCH bound known for cyclic codes; specifically, those bounds on the minimum distance for cyclic codes from defining sets. 

There are dozens of papers on this topic regarding approaches from matrix methods (\cite{BettiSala, ZK1}) through split codes techniques (\cite{HTY,Jensen}) until arriving at the most classic generalizations based on computations over the 
defining set, as the HT bound~\cite{HT}, the Ross bound~\cite{Roos} and the improvements by Van Lint and Wilson \cite{vLW}.

Having so many references on the subject, it seems very necessary to find a general method that allows us to extend any bound for the minimum distance of cyclic codes based on the defining set to the multivariate case. This is our goal. We shall show a method to extend to the multivariate case any bound of the type mentioned via associating an apparent distance to such bound.

To do this, we give in Section 3 a notion of general defining set bound (ds-bound) for the minimum distance of cyclic codes; then, in Section 4 we relate the weight of codewords with the apparent distance of their discrete Fourier transforms. In Section 5, we use this technique to define the apparent distance of an abelian code with respect to a set of ds-bounds. In Section 6, we show an algorithm (of linear complexity by Remark \ref{complejidad}) to compute this apparent distance. Finally, we show how one may improve the searching of new bounds for abelian codes.

\section{Preliminaries}

 Let $\F$  be a finite field with $q$ elements, with $q$ a power of a prime $p$, $r_i$ be positive integers, for all $i\in \{ 1,\ldots, s\}$, and $n=r_1\cdots r_s$.  We denote by $\Z_{r_i}$ the ring of integers modulo $r_i$ and we shall always write its elements as canonical representatives.
 
An \textbf{abelian code} of length $n$ is an ideal in the algebra   $\F(r_1,\ldots, r_s)=\F[X_1,\ldots, X_s]/(X_1^{r_1}-1,\ldots, X_s^{r_s}-1)$ and throughout the work  we assume that this algebra is semisimple; that is, $\gcd (r_i,q)=1$, for all $i\in \{ 1,\dots,s\}$.  Abelian codes are also called multidimensional cyclic codes (see, for example, \cite{Imai}).

The codewords are identified with polynomials $f(X_1,\dots,X_s)$ in which, for each monomial, the degree of the indeterminate $X_k$ belongs to $\Z_{r_k}$. 
We denote by $I$ the set $\Z_{r_1}\times\cdots\times \Z_{r_s}$ and we  write the elements $f \in  \F(r_1,\dots,r_s)$ as 
$f=f(X_1,\dots,X_s)=\sum a_\bf{i} \bf{X}^\bf{i}$, where $\bf{i}=(i_1,\dots, i_s)\in I$ and $\bf{X}^\bf{i}=X_1^{i_1}\cdots X_s^{i_s}$. Given a polynomial $f \in \F[X_1,\dots,X_s]$ we denote by $\overline{f}$ its image under the canonical projection onto $\F(r_1,\dots,r_s)$.

For each $i\in \{ 1,\ldots, s\}$, we denote by $R_{r_i}$ (resp., $U_{r_i}$) the set of all $r_i$-th roots of unity (resp. all $r_i$-th primitive roots of unity) and define $R=\prod_{i=1}^s R_{r_i}$ and $U=\prod_{i=1}^s U_{r_i}$.

For $f=f(X_1,\dots,X_s) \in \F[X_1,\dots,X_s]$ and $\bar{\alpha}\in R$, we write $f(\bar{\alpha})=f(\alpha_1,\dots,\alpha_s)$. For $\bf{i}=(i_1,\ldots, i_s)\in I$, we write
$\bar{\alpha}^{\bf{i}} = (\alpha_1^{i_1}, \dots ,\alpha_s^{i_s})$.

It is a known fact that every abelian code $C$ in $\F(r_1,\dots,r_s)$ is totally determined by its \textbf{root set} or \textbf{set of zeros} 
$$
Z(C)=\left\{\bar{\alpha}\in  R \tq f(\bar{\alpha})=0,\;\; \mbox{ for all }\; f\in C \right\}.
$$ 
The set of non zeros is denoted by $\overline{Z(C)}=\Z_n\setminus Z(C)$.  For a fixed $\bar{\alpha}\in U$, the code $C$ is  determined by its \textbf{defining set}, with respect to $\alpha$, which is defined as 

$$
\D_{\bar{\alpha}}(C) = \left\{ \bf{i}=(i_1,\dots,i_s)\in I \tq f(\bar{\alpha}^{\bf{i}})=0, 
\text{ for all } f\in C\right\}.
$$

Given an element $a=(a_1,\dots,a_s)\in I$, we shall define its \textbf{$q$-orbit} modulo  $\left(r_1,\ldots,r_s \right)$ as $ Q(a)=\left\{\left(a_1\cdot q^{i} ,\dots, a_s\cdot q^{i}  \right) \in I \tq i\in \N\right\}.$ In our case, it is know that the defining set
 $\D_{\bar{\alpha}}\left(C\right)$ is a disjoint union of $q$-orbits modulo $(r_1,\dots,r_s)$. Conversely, every union of $q$-orbits modulo $(r_1,\dots,r_s)$ determines an abelian code (an ideal) in $ \F(r_1,\dots,r_s)$  (see, for example, \cite{BBCS2} for details). We recall that the notions of root set and defining set also apply to polynomials. Moreover, if $C$ is the ideal generated by the polynomial $f$ in $\F(r_1,\dots,r_n)$, then $\D_{\bar{\alpha}}\left(C\right)=\D_{\bar{\alpha}}\left(f\right)$.

We recall that the notion of defining set also applies to cyclic codes. For $s=1$ and  $r_1=n$, a $q$-orbit is called a \textbf{ $q$-cyclotomic coset} of a positive integer $b$ modulo $n$ and  it is the set $C_{q}(b) = \{b\cdot q^{i}\in \Z_{n} \tq i\in \N\}$.

Let $\Le|\F$ be an extension field containing $U_{r_i}$, for all  $i\in \{ 1 \dots,s\}$. The \textbf{discrete Fourier transform of a polynomial $f\in \F(r_1,\dots,r_s)$ with respect to $\bar{\alpha} \in U$} (also called Mattson-Solomon polynomial in \cite{Evans}) is the polynomial  
$\varphi_{\bar{\alpha},f}(\bf{X})=\sum_{\bf{j}\in I} f(\bar{\alpha}^{\bf{j}})\bf{X}^{\bf{j}}\in \Le(r_1,\dots,r_s).$
It is known that the discrete Fourier transform may be viewed as an isomorphism of algebras 
$\varphi_{\bar{\alpha}}:\Le(r_1,\dots,r_s)\longrightarrow (\Le^{|I|},\star),$
 where the multiplication ``$\star$'' in $\Le^{|I|}$ is defined coordinatewise. Thus, we may see $\varphi_{\bar{\alpha},f}$ as a vector in $\Le^{|I|}$ or as a polynomial in $\Le(r_1,\dots,r_s)$ (see \cite[Section 2.2]{Camion}).

\section{Defining set bounds for cyclic codes}

In this section we deal with cyclic codes; that is $r_1=n$.  By $\mathcal{P}(\Z_n )$ we denote the set of the parts of  $\Z_n $. We take an arbitrary $\alpha\in U_n$.\\

\begin{definition} \label{boundg}
A \textbf{defining set bound} (or \textbf{ds-bound}, for short) for the minimum distance  of cyclic codes  is a family of relations $\delta = \{\delta_n\}_{n\in \N} $ such that, for each $n\in \N $, $\delta_n \subseteq \mathcal{P}(\Z_n ) \times \N$ satisfies the following conditions:
\begin{enumerate}
\item If $C$ is a cyclic code in $\mathbb{F}(n)$ such that $\emptyset \neq N \subseteq \D_{\alpha}(C)$, then $1\leq a\leq d(C)$, for all $(N, a) \in \delta_n$.
\item If $\emptyset \neq N \subseteq M$ are subsets of  $\Z_n $ then
$(N, a) \in \delta_n$ implies $(M, a) \in \delta_n$.
\item For all $N \in  \mathcal{P}(\Z_n )$, $(N,1)\in \delta_n$. 
\end{enumerate}
\end{definition}

From now on, sometimes we write simply $\delta$ to denote a ds-bound or any of its elements independently on the length of the code. It will be clear in the context which one is being used.\\

\begin{remarks}
\textbf{(1)} For example, the BCH bound states that for any cyclic code in $\mathbb{F}_q (n)$ that has a string of $t -1$ consecutive powers of some $\alpha  \in U_n$, the minimum distance of the code is at least $t$ \cite[Theorem 7.8]{MWS}. 

Now, define $\delta \subset \mathcal{P}(\Z_n ) \times \N$ as follows: 
for any $a\geq 2$, $(N, a) \in \delta$ if and only if there exist 
$i_0, i_1,\ldots,i_{a-2}$  in $N$ which are consecutive integers modulo $n$. Then the  BCH bound says that $\delta$ is a ds-bound, for any cyclic code (we only have to state Condition 3 as a convention).
 
\textbf{(2)} It is easy to check that all extensions of the BCH bound, all new bounds  from the defining set of a cyclic code as in \cite{BettiSala,HT,Roos1,Roos,ZK1} and the new bounds and improvements arising from Corollary 1, Theorem 5 and results in Section 4 and Section 5 in \cite{vLW}, also verify Definition~\ref{boundg}.\\
\end{remarks}

In order to relate the idea of ds-bound with the Camion's apparent distance, which will be defined later, we consider the following family of maps.\\

\begin{definition}
Let $\delta$ be a ds-bound for the minimum distance of cyclic codes. The  \textbf{optimal ds-bound associated to $\delta$} is the family $\overline{\delta}= \{\overline{\delta}_n\}_{n\in \N} $ of maps $\overline{\delta}_n :  \mathcal{P}(\Z_n ) \longrightarrow  \N$ defined as 
$\overline{\delta}_n(N) = \max \{ b\in \N \,|\, (N,b) \in \delta_n \}$.\\
\end{definition}

The following result is immediate.\\

\begin{lemma} \label{lemma1} 
Let $\delta$ be a ds-bound for the minimum distance of cyclic codes.   Then, for each $n\in \Z$:
\begin{enumerate}
\item If $C$ is a cyclic code in   $\mathbb{F}(n)$ such that $\emptyset \neq N \subseteq \D_{\alpha}(C)$, then $1\leq \overline{\delta}_n (N)\leq d(C)$.

\item If $\emptyset \neq N \subseteq M\subseteq \Z_n $, then
$\overline{\delta}_n (N) \leq \overline{\delta}_n (M)$ . $\blacksquare$\\
\end{enumerate}
\end{lemma}
As we noted above, we may omit the index of the map $\bar{\delta}_n$, because it will be clear on the context for which value it is being taken.

\section{Apparent distance of matrices}

We begin this section recalling the notion and notation of a hypermatrix that will be used hereby, as it is described in \cite{BBCS2}. 
For any $\bf{i}\in I$, we write its $k$-th coordinate as $\bf{i}(k)$. A \textbf{hypermatrix with entries in a set $R$ indexed by $I$ (or an $I$-hypermatrix over $R$)} is an $s$-dimensional $I$-array, denoted by $M=\left(a_{\bf{i}}\right)_{\bf{i}\in I}$, with $a_{\bf{i}}\in R$ \cite{Yamada}. The set of indices, the dimension and the ground field will be omitted if they are clear in the context. For $s=2$, $M$ is a matrix and when $s=1$, $M$ is a vector. We write $M=0$ when all its entries are $0$ and   $M\neq 0$, otherwise.  As usual, a \textbf{hypercolumn} is defined as $H_M(j,k)=\left\{ a_{\bf{i}}\in M\tq \bf{i}(j)=k\right\}$, with $1\leq j\leq s$ and $0\leq k<r_j$, where $a_{\bf{i}}\in M$ means that $a_{\bf{i}}$ is an entry of $M$. A hypercolumn will be seen as an $(s-1)$-dimensional hypermatrix. In the case $s=2$, we refer to hypercolumns as rows or columns and, when $s=1$, we say entries.

Let $D\subseteq I$. The \textbf{hypermatrix afforded by $D$} is defined as $M=\left(a_{\bf{i}}\right)_{\bf{i}\in I}$, where $a_{\bf{i}}=1$ if $\bf{i}\not\in D$ and $a_{\bf{i}}=0$, otherwise. When $D$ is an union of $q$-orbits we say that $M$ is a \textbf{$q$-orbits hypermatrix}, and it will be denoted by $M=M(D)$. For any $I$-hypermatrix $M$ with entries in a ring, we define the support of $M$ as the set $\supp(M)=\left\{ \bf{i}\in I \tq a_{\bf{i}}\neq 0\right\}$. Its complement with respect to $I$ will be denoted by $\D(M)$. 
To define and compute the apparent distance of an abelian code we will use the hypermatrix afforded by its defining set, with respect to $\bar{\alpha} \in U$.

 We define a partial ordering over the set $\left\{M(D)\tq D\text{ is union of $q$-orbits of $I$}\right\}$ of hypermatrices   as follows:
\begin{equation}\label{matrixordering}
M(D)\leq M(D') \Leftrightarrow \supp\left(M(D)\right)\subseteq \supp\left(M(D')\right).
\end{equation}
Clearly, this condition is equivalent to $D'\subseteq D$. We begin with the apparent distance of a vector in $\Le^{n}$.

\ 

\begin{definition} 
Let $\delta$ be a ds-bound for the minimum distance of cyclic codes and $v\in \mathbb{L}^n$ a vector. The \textbf{apparent distance of $v$ with respect to $\delta$} (or \textbf{$\delta$-apparent distance of $v$}, for short), denoted by $\delta^*(v)$, is defined as 
\begin{enumerate}
\item If $v=0$, then $\delta^*(v)=0$.
\item If $v\neq 0$, then $\delta^*(v)= \overline{\delta}(\mathbb{Z}_n\setminus \supp(v))$. 
\end{enumerate}
\end{definition}
From now on we denote by $\B$ a set of ds-bounds that are used to proceed a  computation of the apparent distances of matrices, hypermatrices or abelian codes. 

\

\begin{definition} \label{appdistbound}
Let $v\in \mathbb{L}^n$. The \textbf{apparent distance of $v$ with respect to $\B$}  denoted by $\Delta_\B(v)$, is: 
\begin{enumerate}
\item If $v=0$, then $\Delta_\B(v)=0$.
\item If $v\neq 0$, then $\Delta_\B(v)= \max\{\delta^*(v)\mid \delta\in\B\}$. 
\end{enumerate}
\end{definition}

\

\begin{remarks}\label{propiedades dist apar en 1 var}
 The following properties arise straightforward from the definition above, for any $v\in \mathbb{L}^n$.
 \begin{enumerate}
  \item If $v\neq 0$ then $\Delta_{\B}(v)\geq 1$.
  \item If $\supp (v)\subseteq \supp(w)$ then $\Delta_{\B}(v)\geq \Delta_{\B}(w)$. \\
 \end{enumerate}
\end{remarks}

\begin{proposition}\label{prop1}
Let $f\in \Le(n)$ and $v$ be the vector of its coefficients. Fix any $\alpha\in U_n$. Then $ \Delta_{\B}(v) \leq \omega(\varphi_{\alpha,f}^{-1})$.
\end{proposition}
\begin{proof}
  Set $N=\mathbb{Z}_n\setminus \supp(v)$ and let  $C$ be the abelian code generated by $\varphi_{\alpha,f}^{-1} $ in $\mathbb{L}(n)$. Then $d(C)\leq \omega(\varphi_{\alpha,f}^{-1})$.  By properties of the discrete Fourier transform we have $N=\D_{\bar{\alpha}}(\varphi_{\alpha,f}^{-1} )=\D_{\bar{\alpha}}(C)$ and then, by Lemma~\ref{lemma1} and the definition of apparent distance, $\Delta_{\B}(v)\leq d(C)$. This gives the desired inequality.
\end{proof}

\

Now we define the apparent distance of matrices and hypermatrices with respect to a set $\B$ of ds-bounds.

\

\begin{definition} 
Let $M$ be an $s$-dimensional $I$-hypermatrix over a field $\mathbb{L}$. The \textbf{apparent distance of $M$ with respect to $\B$}, denoted by $\Delta_{\B}(M)$,  is defined as follows:

\begin{enumerate}
 \item $\Delta_{\B}(0)=0$ and, for $s=1$,  Definition~\ref{appdistbound} applies.
\item For $s=2$ and a nonzero matrix $M$, note that $H_M(1,i)$ is the $i$-th row  and $H_M(2,j)$ is the $j$-th column of $M$. Define the \textbf{row support of $M$} as $\supp_1(M)= \{ i\in \{0, \ldots r_1-1\}  \,|\, H_M(1,i)\neq 0 \}$ and the \textbf{column support of $M$} 
as $\supp_2(M)= \{ j\in \{0, \ldots r_2-1\}  \,|\, H_M(2,j)\neq 0 \} $.

Then put
\begin{eqnarray*}
 \omega_{2}(M)&=& \max\{\overline{\delta}(\mathbb{Z}_{r_2}\setminus \supp_2(M))\mid \delta\in\B\},\\
 \epsilon_{2}(M)&=&\max \{ \Delta_\B(H_M(2,j)) \tq j\in \supp_2(M)\}
\end{eqnarray*}
and set $\Delta_{2}(M)=\omega_{2}(M)\cdot \epsilon_{2}(M)$.

Analogously, we compute the apparent distance $\Delta_1(M)$ for the other  variable and finally we define the \textbf{apparent distance of $M$ with respect to $\B$} by
\[
\Delta_{\B}(M) = \max \{\Delta_1(M),\Delta_{2}(M)\}.
\]

\item For $s>2$, proceed as follows: suppose that one knows how  to compute the apparent distance 
$\Delta_{\B}(N)$, for all non zero  hypermatrices of dimension $s-1$. Then first compute the ``hypermatrix support'' of $M\neq 0$ with respect to the $j$-th hypercolumn, that is, $\supp_j(M) = \{ i\in 
\{0, \ldots r_j-1\} \,|\, H_M(j,i)\neq 0\}$. Now put
$ \omega_{j}(M)= \max\{\overline{\delta}(\mathbb{Z}_{r_j}\setminus \supp_j(M))\mid \delta\in \B\}$ and $ \epsilon_{j}(M)=\max \{ \Delta_{\B}(H_M(j,k)) \,|\,k\in \supp_j(M)\}$
and set $\Delta_j(M)=\omega_{j}(M)\cdot \epsilon_{j}(M)$. 

Finally, define the \textbf{apparent distance of $M$ with respect to $\B$} (or the $\B$-apparent distance) as:
$$\Delta_{\B}(M)= \max \left\{ \Delta_j(M)\tq j\in\{1,\dots,s\}\right\}.$$
 \end{enumerate}
\end{definition}

For example, by taking $\B=\{\delta_{BCH}\}$, $\Delta_{\B}(M)$ is the strong apparent distance in \cite{BBCS2}. On the other hand, we note that condition (2) of Remark~\ref{propiedades dist apar en 1 var} not necessarily holds in two or more variables.

Before showing examples, we relate the apparent distance to weight of codewords. For each polynomial $f=\sum_{\mathbf{i}\in I} a_{\mathbf{i}}\mathbf{X}^{\mathbf{i}}$,  consider the hypermatrix  of the coefficients of $f$, denoted by $M(f) = (a_{\mathbf{i}})_{\mathbf{i}\in I}$. For any $j \in \{ 0, \ldots, s\}$, if we write $f=\sum_{k=0}^{r_j-1}f_{j,k}X_j^k$, where $f_{j,k}=f_k(\mathbf{X}_j)$ and $\mathbf{X}_j=X_1\cdots X_{j-1}\cdot X_{j+1} \cdots X_s$, then $M(f_{j,k})= H_M(j,k)$. This means that ''fixed" the variable $X_j$ in $f$, for each power $k$ of $X_j$, the coefficient $f_{j,k}$ is a polynomial in $\mathbf{X}_j$, and $H_M(j,k)$ is the hypermatrix obtained from the coefficients of this  $f_{j,k}$. Now we extend Proposition~\ref{prop1} to several variables.
\begin{theorem}\label{boundDFTnvar}
Let $f\in \mathbb{L}(r_1,\ldots, r_n)$ and $M(f)$ the hypermatrix of its coefficients. Fix any 
$\bar{\alpha}\in U$.  Then  $\Delta_{\B}(M(f)) \leq \omega\left(\varphi^{-1}_{\overline{\alpha},f}\right)$.
 \end{theorem}
 \begin{proof}
  The case $n=1$ is Proposition~\ref{prop1}. We prove the theorem for matrices. The general case follows directly by induction.
  
  Set $\bar{\alpha}=(\alpha_1,\alpha_2)$ and write $f=\sum_{k=0}^{r_2-1}f_{2,k}X_j^k$. Then $H_M(2,k)$ is the vector of coefficients of $f_{2,k}$. Clearly, $\supp_2(M)=\supp(f_{2,k})$ and, for any $k\in \supp_2(M)$, we have 
	$\Delta_{\B}\left(H_M(2,k)\right)\geq 1$. Then, by Proposition~\ref{prop1}, $\omega\left(\varphi^{-1}_{\alpha_1,f_{2,k}}\right)\geq 1$.
  
  By the definition of discrete Fourier transform $ |\overline{Z(f_{2,k})} | = \omega(\varphi_{\alpha_1,f_{2,k}}^{-1})$, hence there exists $t\in \Z_{r_1}$ (at least one) such that $\alpha_1^t$ is a non zero of $f_{2,k}$.
  
  Set $ g(X_2)=f(\alpha_1^t,X_2)=\sum_{k=0}^{r_2-1}f_{k}(\alpha_1^t)X_2^k$ and note that it is a non zero polynomial. Let $v(g)$ be the vector of coefficients of $g(X_2)$. Then $\supp\left(v(g)\right)\subseteq \supp_2(M)$ and so $\Delta_{\B}(v(g))\geq \max\{\bar{\delta}\left(\Z_{r_2}\setminus \supp_2(M)\right)\tq \delta\in \B\}=\omega_2(M)$.
  
  As $|\overline{Z(g)}|=\omega(\varphi_{\alpha_2,g}^{-1})\geq \Delta_{\B}(v(g))$ then $|\overline{Z(f)}|\geq \omega_2(M)\cdot \Delta_{\B}\left(H_M(2,k)\right)$. Finally, as in the univariate case, it is clear that $ |\overline{Z(f)} | = \omega(\varphi_{\bar{\alpha},f}^{-1})$ and this completes the proof.
 \end{proof}\\

\begin{example}\label{Ex 3por24}
Set $n=72=3\times 24$ and $q=5$. Fix $\alpha_1\in U_{3}$ and $\alpha_2\in U_{24}$ and consider the $5$-orbits matrix $M$ afforded by $D=Q(0,0)\cup Q(0,1)\cup Q(0,2)\cup Q(0,3) \cup  Q(0,6)\cup  Q(0,7) \cup Q(0,9) \cup Q(1,0) \cup  Q(1,1) \cup  Q(1,5) \cup  Q(1,6).$ Choose $\B=\{\delta_{HT},\delta_{BS}\}$, where $\delta_{HT}$ is the Hartman-Tzeng bound in \cite{Roos} (see \cite{HT}) and $\delta_{BS}$ is the Betti-Sala bound in \cite{BettiSala}. One may check that $\supp_1(M)=\{0,1,2\}$ and $\omega_1(M)=1$. On the other hand, $\supp_2(M)=\Z_{24}\setminus \{0,1,5,6\}$ and $\omega_2(M)=4$, by using $\delta_{HT}$. Now $\Delta_{\B}(H_M(1,0))=8$, by using $\delta_{BS}$ (see \cite[Example 4.2]{BettiSala}) which is the maximum, and hence $\epsilon_1(M)=8$. Finally it is clear that $\epsilon_2(M)=2$, so that $\Delta_1(M)=\Delta_2(M)=8$. We conclude $\Delta_{\{\delta_{HT},\delta_{BS}\}}(M)=8$.

The reader may check that $\Delta_{\delta_{HT}}(M)=8$ too, because although $\Delta_1(M)=6$ we get $\Delta_2(M)=8$. Moreover, one may check that for $\B=\{\delta_{BCH}\}$ we get $\Delta_1(M)=5$ and $\Delta_2(M)=6$ (because $\omega_2(M)=3$), hence $\Delta_{\B}(M)=6$.\\
\end{example}

\begin{example}\label{Ex 5por15}
 Set $q=2$, $n=75=5\times 15$ and let $M$ be the $q$-orbits matrix afforded by $D=Q(0,0)\cup Q(0,3)\cup Q(0,5)\cup Q(0,7)\cup Q(1,0)\cup Q(1,2)\cup Q(1,4)$.  In this case, one may check that $\Delta_{\{\delta_{BCH},\delta_{HT}\}}(M)=8$.
\end{example}

\section{Apparent distance of abelian codes}
The following definition changes a little the usual way to present the notions of apparent distance in \cite{Camion} and the strong apparent distance in \cite{BBCS2} (see also \cite{Evans}). We recall that $\B$ denotes a set of ds-bounds, which are used to proceed a concrete computation of the apparent distances.

\begin{definition} \label{appdistcodenvar}
Let $C$ be an abelian code in $\F(r_1,\dots,r_s)$.
\textit{1)} The \textbf{apparent distance of $C$ with respect to  $\bar{\alpha}\in U$ and $\B$} (or the $(\B,\bar{\alpha})$-apparent distance) is
$$  
\Delta_{\B,\overline{\alpha}}(C)= \min \{ \Delta_{\B} (M\left(\varphi_{\bar{\alpha},c})\right) \,|\, c\in C\}.
$$

\textit{2)} The \textbf{apparent distance of $C$ with respect to  $\mathbb{B}$} is
 $$\Delta_{\mathbb{B}}(C)= \max \{ \Delta_{\B,\bar{\alpha}}(C) \,|\, \bar{\alpha}\in U\}.$$
\end{definition}

The following result is consequence of Theorem~\ref{boundDFTnvar}.\\

\begin{corollary} \label{corTheo}
For any abelian code $C$ in $\F(r_1,\dots,r_s)$ and any $\B$ as above, $\Delta_{\mathbb{B}}(C) \leq d(C)$.
\end{corollary}
\begin{proof}
 Let $g\in C$ such that $\omega(g)=d(C)$. By Theorem~\ref{boundDFTnvar}, $\Delta_{\B}\left(M\left(\varphi_{\bar{\alpha},g}\right)\right)\leq \omega(g)$ for any $\alpha\in U$. From this the result follows directly.
\end{proof}\\

It is certain that to compute the apparent distance for each element of a code in order to obtain the apparent distance of the code can be as hard work as to compute the minimum distance of such a code. Hence, to improve the efficiency of the computation the following result tells us that we may restrict our attention to the idempotents of the code. It also allows us to reformulate Definition \ref{appdistcodenvar} as it is presented in the mentioned papers.

 \begin{proposition} \label{prop2}
Let $C$ be an abelian code in $\F(r_1,\dots,r_s)$. The apparent distance of $C$ with respect to  $\bar{\alpha}\in U$ and $\B$ verifies the equality
$$ \Delta_{\B, \bar{\alpha}}(C) = \min \{ \Delta_{\B} \left(M(\varphi_{\bar{\alpha},e})\right) \,|\, e^2=e\in C\}.$$
\end{proposition}
\begin{proof}
 Consider any $c\in C$. Since $\F(r_1,\dots,r_s)$ is semisimple, then there exists an idempotent $e\in C$ such that the ideals generated by $c$ and $e$ in $\F(r_1,\dots,r_s)$ coincide; that is, $(c)=(e)$, and so, $\D_{\bar{\alpha}}(c)=\D_{\bar{\alpha}}(e)$. This means that $\supp\left(M(\varphi_{\bar{\alpha},c})\right) =\supp\left(M(\varphi_{\bar{\alpha},e})\right)$. Now, one may see that the computation of the apparent distance is based on the fact that the entries (of the matrices) are zero or not; that is, once an entry is non zero, its specific value is irrelevant. 
Considering this fact, it is easy to see that $ \Delta_{\B} \left(M(\varphi_{\bar{\alpha},c})\right)= \Delta_{\B} \left(M(\varphi_{\bar{\alpha},e})\right)$ and so we get the desired equality.
\end{proof}\\

If $e\in\F(r_1,\dots,r_s)$ is an idempotent and $E$ is the ideal generated by $e$ then $\varphi_{\bar{\alpha}, e}\star\varphi_{\bar{\alpha}, e}=\varphi_{\bar{\alpha}, e}$, for any $\bar{\alpha} \in U$.  Thus if $\varphi_{\bar{\alpha}, e}=\sum_{\bf{i}\in I} a_{\bf{i}}X^{\bf{i}}$, we have $a_{\bf{i}}\in \{1,0\}\subseteq \F$ and $a_{\bf{i}}=0$ if and only if $\bf{i}\in \D_{\bar{\alpha}}(E)$. Hence $M(\varphi_{\bar{\alpha},e})=M(\D_{\bar{\alpha}}(E))$. Conversely, let $M$ be a  hypermatrix afforded by a set $D$ which is a union of $q$-orbits. We know that $D$ determines a unique ideal $C$ in $\F(r_1,\dots,r_s)$ such that $\D_{\bar{\alpha}}(C)=D$. Let $e\in C$ be its generating idempotent. One may verify that  $M(\varphi_{\bar{\alpha}, e})=M(D)$.

Now let $C$ be an abelian code, $\bar{\alpha} \in U$ and let $M$ be the hypermatrix afforded by $\D_{\bar{\alpha}}(C)$. For any $q$-orbits hypermatrix $P\leq M$ [see the ordering (\ref{matrixordering})] there exists a unique idempotent $e'\in C$ such that $P=M(\varphi_{\bar{\alpha}, e'})$ and for any codeword $f\in C$ there is a unique idempotent $e(f)$ such that $\Delta_{\B}\left(M(\varphi_{\bar{\alpha},f})\right)= \Delta_{\B} \left(M(\varphi_{\bar{\alpha},e(f)})\right)$. So
\begin{eqnarray*}
  \min\{\Delta_{\B}(P)\tq 0\neq P\leq M\}= \\ \min\{\Delta_{\B}(M(\varphi_{\bar{\alpha}, e}))\tq
  0\neq e^2=e\in C\}= \Delta_{\B}(C).
\end{eqnarray*}
This fact drives us to give the following definition.
\begin{definition}
 For a $q$-orbits hypermatrix $M$, its \textbf{minimum $\mathbb{B}$-apparent distance} is
	\[\mathbb{B}\mad(M)=\min\{\Delta_{\B}(P)\tq 0\neq P\leq M\}.\]
\end{definition}

Finally, in the next theorem we set the relationship between the apparent distance of an abelian code and the minimum apparent distance of hypermatrices.

\begin{theorem}
Let $C$ be an abelian code in $\F(r_1,\dots,r_s)$ and let $e$ be its generating idempotent. For any $\bar{\alpha}\in U$ we have $\Delta_{\B,\bar{\alpha}}(C)= \mathbb{B}\mad\left(M(\varphi_{\bar{\alpha},e})\right)$. Therefore, $\Delta_{\mathbb{B}} (C)= \max\{\mathbb{B}\mad\left(M(\varphi_{\bar{\alpha},e})\right)\tq \bar{\alpha}\in U\}$.
\end{theorem}
\begin{proof}
It follows  directly from the preceding  paragraphs.
\end{proof}

\section{Computing minimum apparent distance}

In \cite{BBCS2} it is presented an algorithm to find, for any abelian code, a list of matrices representing some of its idempotents (or hypermatrices in case of more than 2 variables) whose apparent distance based on the BCH bound (called the strong apparent distance) goes decreasing until the minimum value is reached. It is a kind of ``suitable idempotents chase through hypermatrices'' \cite[p. 2]{BBCS2}.
This algorithm is based on certain manipulation of the hypermatrix afforded by the defining set of the abelian code. 
It is not so hard to see that it is possible to obtain an analogous algorithm in our case.


We reproduce here the result and algorithm in the case of two variables under our notation. Then we will use the mentioned algorithm to improve the searching of new bounds for abelian codes.

\

\begin{definition}\label{indices involucradas}
 With the notation of the previous section, let $D$ be a union of $q$-orbits and $M=M(D)$ the hypermatrix afforded by $D$. We say that $H_M(j,k)$ is an \textbf{involved hypercolumn (row or column for two variables) in the computation of $\Delta_{\B}(M)$}, if $\Delta_{\B}(H_M(j,k))=\epsilon_{j}(M)$ and $\Delta_{j}(M)=\Delta_{\B}(M)$.
\end{definition}

\

 We denote the set of indices of involved hypercolumns by $I_p(M)$.
Note that the involved hypercolumns are those which contribute in the computation of the $\B$-apparent distance.

The next result shows a sufficient condition to get at once the minimum $\mathbb{B}$-apparent distance of a hypermatrix.

\

\begin{proposition}\label{matrizdamvarias}
With the notation as above, let $D$ be a union of $q$-orbits and $M=M(D)$ the hypermatrix afforded by $D$. If $\Delta_{\B}(H_M(j,k))=1$, for some $H_M(j,k)\in I_p(M)$, then $\mathbb{B}\mad(M)=\Delta_{\B}(M)$. 
\end{proposition}
\begin{proof}
 It is a modification of that in \cite[Proposition 23]{BBCS2} having in mind the use of different ds-bounds.
\end{proof}\\

\begin{theorem} \label{teodam2}
Let $\Q$ be the set of all $q$-orbits modulo $(r_1,r_2)$, $\mu\in \{1,\dots, |\Q|-1\}$ and $\left\{Q_j\right\}_{j=1}^{\mu} $ a subset of $\Q$. Set $D=\cup_{j=1}^\mu Q_j$ and $M=M(D)$. Then there exist two sequences: the first one is formed by nonzero $q$-orbits matrices, $M=M_0 > \dots > M_l\neq 0$ and the second one is formed by positive integers
$m_0 \geq  \dots \geq m_l, $ with $ l\leq \mu$ and $m_i\leq \Delta_{\B}(M_i)$, verifying the following property: 

If $P$ is a $q$-orbits matrix such that $0 \neq P \leq M$, then  $\Delta_{\B}(P) \geq m_l$ and if $\Delta_{\B}(P) < m_{i-1}$ then $P \leq M_i$, where $0 < i\leq l$.

 Moreover,  if $l' \in \{0, \dots , l\}$ is the first element satisfying $m_{l'}=m_l$ then $\Delta_{\B}(M_{l'})=\mathbb{B}\mad(M)$. 
\end{theorem}
\begin{proof}
 It follows the same lines of that in \cite[Proposition 25]{BBCS2} having in mind the use of different ds-bounds.
\end{proof} \\

\noindent\textbf{Algorithm for matrices.}   

Set $I=\Z_{r_1}\times\Z_{r_2}$. Consider the matrix $M=\left(a_{ij}\right)_{(i,j)\in I}$ and a set $\B$ of ds-bounds.
\begin{itemize}
\item[]Step 1. Compute the apparent distance of $M$ with respect to $\B$ and set $m_0= \Delta_{\B}(M)$.

\item[]Step 2. 
  \begin{itemize}
      \item[a)]  If there exists $H_M(j,k)\in Ip(M)$ (see Definition~\ref{indices involucradas}) such that $\Delta_{\B}(H_M(j,k))=1$ then we finish giving the sequences $M=M_0$ and $m_0=\Delta_{\B}(M)$ (because of Proposition~\ref{matrizdamvarias}).
      \item[b)]  If $\Delta_{\B}(H_M(j,k)) \neq 1$ for all $H_M(k,b)\in Ip(M)$, we set $$S=\bigcup_{H_M(k,b)\in Ip(M)}supp (H_M(k,b))$$ and construct the matrix $M_1=\left(a_{ij}\right)_{(i,j)\in I}$ such that 
      \[a_{ij}=
      \begin{cases} 
      0 & \text{if}\; (i,j)\in \cup \{Q(k,b)\tq (k,b)\in S\} \\
      m_{ij} & \text{otherwise.}  
    \end{cases}\]
  \end{itemize}
(In other words, $M_1<M$ is the matrix with maximum support such that its involved rows and columns are zero. One may prove that if $0\neq P <M$ and $\Delta_{\B}(P)< m_0$ then $P \leq M_1$.)

\item[]Step 3.
  \begin{itemize}
      \item [a)] If $M_1=0$ then we finish giving the sequences $M=M_0$ and $m_0= \Delta_{\B}(M)$.
      \item [b)] If $M_1 \neq 0$, we set $m_1= \min \{m_0, \Delta_{\B}(M_1)\}$, and we get the sequences $M=M_0>M_1$ and $m_0\geq m_1$. Then, we go back to Step 1 with  $M_1$ in the place of $M$ and $m_1$ in the place  of $m_0$. $\blacksquare$\\
  \end{itemize}
 \end{itemize}
\begin{remark}\label{complejidad}
If the matrix has $\mu$ $q$-orbits the algorithm has at most $\mu$ steps. $\blacksquare$
\end{remark}

\section{Examples}

\begin{example}
We take the setting of Example~\ref{Ex 3por24} and consider the abelian code $C$ with $\D(C)=D$. In this case, the matrix $M=M(\D(C))$ is the same as that in the mentioned example. Choose again $\B=\{\delta_{HT},\delta_{BS}\}$.

By Example~\ref{Ex 3por24}, $m_0=8$. Now $I_p(M)=\{H_M(1,0)\}\cup\{H_M(2,j)\tq j\in \{2,3,7,9, 10, 11, 15,21\}\}$. One may check that $\Delta_{\B}(M_1)=\Delta_1(M_1)=\Delta_2(M_1)=16$. Finally, as $\{H_M(1,1),H_M(1,2)\}\subset I_p(M_1)$ then $M_2=0$.

This code has $\dim_{\mathbb{F}_5}(C)=52$ and $\Delta_{\B}(C)=8$. The closer code we know is a $(105,51,7)$ binary cyclic code in \cite[Table II]{HTY}. The known bounds of linear codes are between 10-15.
\end{example}

\begin{example}
Consider the abelian code related with the matrix in Example~\ref{Ex 5por15}. So
 $q=2$, $n=75=5\times 15$. Let $C$ be the abelian code with defining set $\D(C)=D$, and choose, as in the mentioned example, $\B=\{\delta_{BCH},\delta_{HT}\}$. In this case, $\Delta_{\B}(M_0)=8$,  $\Delta_{\B}(M_1)=8$  and $\Delta_{\B}(M_2)=15$.
 
 This code has $\dim_{\mathbb F_2}C=52$ and $\Delta_{\B}(C)=8$. It has the largest known bound to the minimum distance among the linear codes with the same length and dimension (see \cite[p. 670]{BVtabla} or \cite{Gtabla}). \\
\end{example}

In Table~\ref{tabla} we include binary abelian codes of lengths $n=3\times r_2$ that may be of interest. The codes $C_1$, $C_2$ and $C _5$ have the largest known bound to the minimum distance among the \textit{linear codes} with the same length and dimension (see \cite{BVtabla} or \cite{Gtabla}). The codes $C_3$ and $C_4$ have a bound larger than that of any cyclic code with the same length and dimension (see \cite{ZJtabla}). The ``shifting bound'', denoted $\delta_{SB}$, is considered for the code $C_4$  (see \cite[Example 7]{vLW}). The reader may compare this codes with those in \cite{ZK1} and \cite{Jensen}.
\begin{table}[!t]
\renewcommand{\arraystretch}{1.3}
\caption{Binary abelian codes with $n=3\times r_2$ and big rate}
\label{tabla}
\centering
\begin{tabular}{|c||c||c||p{2cm}||c||c|}

\hline
Code&$n$ & $\dim$& $\D(C)$ representatives& $\B$& $\Delta$\\
\hline
$C_1$&$21$ & 16 & $(0,1),(1,0)$& $\delta_{BCH}$ & 3\\
\hline
$C_2$&$45$ & 39 & $(0,1),(1,0)$& $\delta_{BCH}$ & 3\\
\hline
$C_3$&$51$ & 35 & $(0,1),(1,3)$&  $\delta_{HT}$ & 3\\
\hline
$C_4$&$69$ & 46 & $(0,0),(1,1)$& $\delta_{SB},\,\delta_{BCH}$ & 6\\
\hline
$C_5$&$105$ & 93 & $(0,5)$, $(0,7)$, \newline $(0,15)$,$(1,0)$& $\delta_{HT},\,\delta_{BCH}$ & 8\\
\hline
\end{tabular}
\end{table}

\section{Conclusion}
We have developed a technique to extend any ds-bound for cyclic codes to multivariate codes which can be applied to codes of arbitrary length, mainly for those whose minimum distance is still unknown. We use this technique to improve the searching of new bounds for abelian codes.

%




\end{document}